\RequirePackage{amsmath}
\documentclass[runningheads]{llncs}
\usepackage{graphicx,amsmath,amssymb,bbm,mathtools}
\usepackage{ulem}
\usepackage[dvipsnames]{xcolor}

\newcommand{\eps}{\varepsilon}
\newcommand{\bI}{\mathbbm{1}}

\newcommand{\bR}{\mathbb{R}}
\newcommand{\bZ}{\mathbb{Z}}
\newcommand{\bE}{\mathbb{E}}
\newcommand{\bbI}{\mathbb{I}}
\newcommand{\ALG}{\mathsf{ALG}}
\newcommand{\OPT}{\mathsf{OPT}}
\newcommand{\acc}{\mathsf{acc}}

\newcommand{\cF}{\mathcal{F}}
\newcommand{\cI}{\mathcal{I}}

\newcommand{\tV}{\tilde{V}}
\newcommand{\cY}{\overline{Y}}
\newcommand{\cv}{\mathbf{v}}
\newcommand{\cJ}{\mathbf{J}}
\newcommand{\bx}{\mathbf{x}}
\newcommand{\bv}{\mathbf{v}}

\newcommand{\rank}{\mathrm{rank}}

\newcommand{\ssx}{x^*}

\newcommand{\ssy}{y^*}
\newcommand{\ssD}{D^*}

\newcommand{\hu}{\hat{u}}

\newcommand{\sumt}{\sum_{t=1}^T}
\newcommand{\sumS}{\sum_{S\subseteq N_t}}
\newcommand{\sumk}{\sum_{k=1}^K}

\newcommand{\hI}{\hat{I}}
\newcommand{\hD}{\hat{D}}
\newcommand{\hR}{\hat{R}}
\newcommand{\hV}{\hat{V}}
\newcommand{\vhV}{\mathbf{\hat{V}}}
\newcommand{\Feas}{\mathsf{Feas}}
\newcommand{\proph}{\mathsf{Alloc^*}}

\DeclareMathOperator*{\argmax}{arg\,max}

\begin{document}
\title{Prophet Inequalities on the Intersection of a Matroid and a Graph
\thanks{The authors would like to thank the anonymous reviewers for the 12th \textit{Symposium on Algorithmic Game Theory} (SAGT) who made many suggestions that improved the paper.  This is the full version of a one-page abstract from those proceedings.}
}
\author{Jackie Baek\inst{1}\orcidID{0000-0001-5538-509X} \and
Will Ma\inst{2}\orcidID{0000-0002-2420-4468}
}
\authorrunning{Baek and Ma}
\institute{
Operations Research Center, Massachusetts Institute of Technology, Camridge MA 02139, USA \email{baek@mit.edu} \and
%Operations Research Team,
Operations Research Team, Google Research, Camridge MA 02139, USA
\email{willma@google.com}
}
\maketitle              % typeset the header of the contribution
\begin{abstract}
We consider prophet inequalities in a setting where agents correspond to both elements in a matroid and vertices in a graph.
A set of agents is feasible if they form both an independent set in the matroid and an independent set in the graph.
Our main result is an ex-ante $\frac{1}{2(d+1)}$-prophet inequality, where $d$ is a graph parameter upper-bounded by the maximum size of an independent set in the neighborhood of any vertex.

\hspace{12pt} We establish this result through a framework that sets both dynamic prices for elements in the matroid (using the method of balanced thresholds), and static but discriminatory prices for vertices in the graph (motivated by recent developments in approximate dynamic programming).
The threshold for accepting an agent is then the sum of these two prices.

\hspace{12pt} We show that for graphs induced by a certain family of interval-scheduling constraints, the value of $d$ is 1.
Our framework thus provides the first constant-factor prophet inequality when there are both matroid-independence constraints and interval-scheduling constraints.
It also unifies and improves several results from the literature, leading to a $\frac{1}{2}$-prophet inequality
when agents have XOS valuation functions over a set of items and use them for a finite interval duration, and more generally, a $\frac{1}{d+1}$-prophet inequality when these items each require a bundle of $d$ resources to procure.

\end{abstract}
\section{Introduction}

Prophet inequalities analyze the performance of online vs.\ offline algorithms in sequential selection problems, and have enjoyed a recent surge of uses in posted-price mechanism design.
The typical online selection problem can be described as follows.

A set of $T$ agents is denoted by $N=\{1,\ldots,T\}$.
Each agent $t$ has a valuation $V_t$ drawn from a known distribution $F_t$.
The valuations are realized independently, and revealed sequentially.
Each agent must be irrevocably accepted or rejected upon her valuation being revealed, with the feasibility constraint that the set of agents accepted by the end must lie in $\cF$, a downward-closed collection of subsets of $N$.
The objective is to maximize the expected sum of valuations of agents accepted.
We will refer to this as the welfare.

The algorithm's expected welfare is compared to that of a clairvoyant who can see all the realized valuations beforehand and make ``prophetic'' accept/reject decisions.
All of our results also hold relative to the stronger \textit{ex-ante} prophet, who can choose the correlation between the marginal distributions $F_1,\ldots,F_T$ to maximize his welfare (but for the algorithm, the valuations are still independent).
We let $\OPT$ denote the prophet's expected welfare, which equals $\bE[\max_{S\in\cF}\sum_{t\in S}V_t]$.

In this paper, we analyze the structure where $\cF$ is defined by the intersection of a matroid and a graph.
Specifically, there is a matroid $M=(N,\cI)$ and an undirected graph $G=(N,E)$, both defined on the set of agents $N$.
$\cF$ then consists of the subsets $S\subseteq N$ that are both independent in the matroid, i.e.\ $S\in\cI$, and independent in the graph, i.e.\ $\{t,t'\}\notin E$ for all $t,t'\in S$.

To state our main result, we need the following definitions.

\begin{definition} \label{def::d1}
For a matroid $M=(N,\cI)$, define $d_1(M)=0$ if $\cI=2^N$, and $d_1(M)=1$ otherwise.
\end{definition}
$\cI=2^N$ corresponds to the free matroid, under which all subsets are independent.

\begin{definition} \label{def::d2}
For a graph $G=(N,E)$, define
\begin{align} \label{eqn::dGraphDef}
d_2(G)=\max_{t\in N}\alpha\Big(G\big[\{t'<t:\{t,t'\}\in E\}\big]\Big),
\end{align}
where $G[\cdot]$ denotes the subgraph of $G$ induced by a set of vertices, and $\alpha(\cdot)$ denotes the maximum size of an independent set in a graph.
\end{definition}
We explain expression~(\ref{eqn::dGraphDef}).  $\{t,t'\}\in E$ implies that $t$ cannot be accepted alongside $t'$, and $t'<t$ implies that $t'$ could have been accepted before $t$ to ``block'' agent $t$.
However, some of these agents $t'$ may also block each other, in which case they are adjacent in the induced subgraph $G[\cdot]$.
$\alpha(\cdot)$ counts the maximum number of such agents that can be simultaneously accepted, and $d_2(G)$ takes the maximum of these numbers over $t\in N$.
We note that $d_2(G)$ is upper-bounded by $\max_t\alpha(G[\{t':\{t,t'\}\in E\}])$, the
maximum size of an independent set in the neighborhood of any vertex.

\begin{theorem} \label{thm::mr}
For any matroid $M$ and graph $G$, the expected welfare of an online algorithm is at least
\begin{align*}
\frac{1}{(d_1(M)+1)(d_2(G)+1)}\cdot\OPT.
\end{align*}
\end{theorem}
Our algorithm is \textit{order-aware}, in that it needs to assume the agents' valuations will be revealed in the given order $1,\ldots,T$.
Before elaborating on our techniques in Section~\ref{sec::techniques}, we outline the implications of our Theorem~\ref{thm::mr} and various generalizations relative to the literature, and describe settings where $d_2(G)$ is small.

\subsection{Our Results, in relation to Previous Results} \label{sec::results}

In Theorem~\ref{thm::mr}, $d_2(G)$ is small if an agent cannot be blocked by \textit{many agents that don't block each other}.
One setting where this arises is when the agents arrive in order $1,\ldots,T$,
each requesting service for a duration starting with her time of arrival,
and need to be served by a single server.
Formally, associated with the agents are intervals $\bbI_1=[\ell_1,u_1],\ldots,\bbI_T=[\ell_T,u_T]$ satisfying $\ell_1\le\ldots\le\ell_T$, and a set of agents $S$ can be feasibly served if
\begin{align} \label{eqn::intervalGraph}
|\{t\in S:z\in\bbI_t\}|\le1 && \forall\ z\in[\ell_1, \max_tu_t].
\end{align}
In the graph $G$ induced by constraints~(\ref{eqn::intervalGraph}), two agents are adjacent if their intervals overlap.
For an agent $t$, any agents $t'<t$ with $\bbI_{t'}\cap\bbI_t\neq\emptyset$ must have $\bbI_{t'}$ contain the point $\ell_t$, since $\ell_{t'}\le\ell_t$ and the intervals are contiguous.
Therefore, all of these agents $t'$ are also adjacent to each other in $G$ through the point $\ell_t$, which implies that $d_2(G)\le1$.

We contrast this with a different type of interval constraint where the agents request service starting from a common point in time $\ell$, and there is a time-dependent service capacity $B(z)\in\bZ_{\ge0}$ for all $z\ge\ell$.
The agents request intervals $\bbI'_1=[\ell,u'_1],\ldots,\bbI'_T=[\ell,u'_T]$,
and a set of agents $S$ can be feasibly served if
\begin{align} \label{eqn::intervalMatroid}
|\{t\in S:z\in\bbI_t\}|\le B(z) && \forall\ z\ge\ell.
\end{align}
In (\ref{eqn::intervalMatroid}), since the intervals starting from the same point are nested, the constraints can be captured by a laminar matroid $M$, with $d_1(M)\le1$.

Therefore, Theorem~\ref{thm::mr} shows that the guarantee relative to the ex-ante prophet is at least $\frac{1}{(d_1(M)+1)(d_2(G)+1)}\ge1/4$
under the combination of constraints~(\ref{eqn::intervalGraph}) and~(\ref{eqn::intervalMatroid}).
This could model an
online rectangle packing problem, where the
horizontal projections have increasing left-boundaries and must satisfy (\ref{eqn::intervalGraph}), while the
vertical projections have identical top-boundaries and must satisfy (\ref{eqn::intervalMatroid}).
More generally, Theorem~\ref{thm::mr} implies a $(1/4)$-guarantee for any online matroid selection problem
under the additional constraint that each agent requires a processing time, during which no other agent can be served even if they are independent in the matroid.
To our knowledge, our framework provides the first constant-factor guarantee under the combined families of feasibility constraints.
Indeed, the constraints~(\ref{eqn::intervalGraph}) do not correspond to a matroid.\footnote{
Consider an example where intervals $\bbI_2,\ldots,\bbI_T$ are disjoint, but interval $\bbI_1$ is long and overlaps all of them.
In this case, constraints~(\ref{eqn::intervalGraph}) do not correspond to a matroid, and in fact require the intersection of $\Omega(T)$ matroids to capture.  Hence, applying a result from \cite{kleinberg2012matroid} would not yield a constant-factor guarantee.
}
Meanwhile, (\ref{eqn::intervalMatroid}) cannot be captured by the pairwise independence constraints of a graph.
If $E=\emptyset$ and the graph imposes no feasibility constraints, then $d_2(G)=0$ and the guarantee from Theorem~\ref{thm::mr} is $1/2$, which is the matroid prophet inequality from \cite{kleinberg2012matroid}.

In Section~\ref{sec::generalizations},
we consider the generalized setting studied in \cite{feldman2014combinatorial,duetting2017prophet}, where the agents have XOS valuation functions over a set of items.
We impose matroid- and graph- independence constraints on the subset of items allocated to the agents by the end, and
show that the guarantee of $\frac{1}{(d_1(M)+1)(d_2(G)+1)}$ from Theorem~\ref{thm::mr} still holds (Theorem~\ref{thm::XOS}).
If the matroid is free, then $d_1(M)=0$, and 
a corollary of Theorem~\ref{thm::XOS} is that the $\frac{1}{2}$-guarantee for XOS from \cite{feldman2014combinatorial}
still holds if the agents use the items allocated to them for a finite interval duration (instead of keeping the items forever).
More generally, we show that if each item requires a bundle of at most $d$ underlying resources (possibly for a finite interval duration) to procure, then $d_2(G)\le d$, leading to a guarantee of $1/(d+1)$ (Proposition~\ref{prop::d2<=d}).

\subsection{Our Techniques, in relation to Previous Techniques} \label{sec::techniques}

Central to the development of prophet inequalities is the notion of a \textit{residual function}.
In the basic setting with an arbitrary feasible collection $\cF$, if $Y\in\cF$ is the set of agents that have already been accepted, then its residual is defined as
\begin{align} \label{eqn::originalResidual}
R(Y)=\bE_{\tV_1,\ldots,\tV_T}\left[\max_{S:S\cup Y\in\cF}\sum_{t\in S}\tV_t\right],
\end{align}
where $\tV_1,\ldots,\tV_T$ is a freshly sampled set of valuations.
The algorithm decides whether to accept an agent $t$ by comparing the actual realization of $V_t$ with the simulated threshold of $\gamma(R(Y)-R(Y\cup\{t\}))$, where $\gamma$ is a constant in $(0,1)$.
$\gamma$ is chosen depending\footnote{We note, however, that $\gamma$ does not depend on the given valuation distributions.} on $\cF$ to \textit{balance} the thresholds---for example, if $\cF$ is the independent sets of a matroid, then $\gamma=1/2$ ensures that the thresholds are neither too high nor too low \cite{kleinberg2012matroid}.
In the simplest XOS-valuation setting with only item capacity constraints,
the difference in residuals decomposes very nicely as a sum of \textit{item-prices} \cite{feldman2014combinatorial}.
It is important, however, that these residuals and prices are always computed based on a prophet who
``starts over'' and considers
%the valuations of
every agent $1,\ldots,T$, even when some agents have already come and gone.

Unfortunately, this ``starting over'' does not exploit the temporal aspect of graph-independence constraints, as illustrated by the following example.
\begin{example} \label{ex::KWBad}
$T$ agents arrive in order, with agent 1 requesting service for a long interval $\bbI_1=[1,T+1]$, and each agent $t\ge2$ requesting service for a short interval $\bbI_t=[t,t+1/2]$.
There is a single server, so a set of agents $S$ is feasible if and only if $S$ satisfies constraints~(\ref{eqn::intervalGraph}).
Agent 1 has valuation $\frac{C+T\eps}{\eps}$ with probability $\eps$, and valuation 0 otherwise, where $C,\eps>0$ are constants.
Agents $2,\ldots,T$ deterministically have valuation 1.
\end{example}
Consider the thresholds set by the residual function~(\ref{eqn::originalResidual}) on Example~\ref{ex::KWBad}.
If no agents have been taken, i.e.\ $Y=\emptyset$, then
\begin{align*}
R(\emptyset)=\frac{C+T\eps}{\eps}\cdot\eps+(T-1)(1-\eps)=C+T-1+\eps.
\end{align*}
Meanwhile, taking any agent $t\ge2$ prevents the prophet from taking agent 1, so $R(\{2\})=\ldots=R(\{T\})=T-1$.\footnote{
Note that in our definition~(\ref{eqn::originalResidual}), taking an agent $t$ does not prevent the prophet from taking $t$ again ($S$ and $Y$ need not be disjoint),
although this distinction is purely for convenience and not critical.
}
Therefore, when $Y=\emptyset$, the threshold for agents $t\ge2$ is $\gamma(R(\emptyset)-R(\{t\}))=\gamma(C+\eps)$.
By setting $C=1/\gamma$, we can always construct an instance where these thresholds exceed 1, rejecting all agents $2,\ldots,T$.
The algorithm's welfare is only $\frac{C+T\eps}{\eps}\cdot\eps=C+T\eps$, while the prophet's welfare is $\OPT=R(\emptyset)\approx C+T$.
By choosing $T\gg C$ and $\eps=o(1/T)$,
the algorithm's welfare can be made an arbitrarily small fraction of $\OPT$.

The residual-based approach\footnote{
There are also more extensive frameworks \cite{duetting2017prophet,lucier2017economic}, but to the best of our knowledge, they still perform poorly on Example~\ref{ex::KWBad}, as we discuss in Appendix~\ref{appx::balancedPricesDNE}.
}
performed poorly on Example~\ref{ex::KWBad} because the first agent's existence continued to inflate the thresholds of agents $2,\ldots,T$, even after she had
already
come and gone.
To improve upon it, we incorporate dynamic programming, which is particularly designed to account for these temporal dynamics.
Motivated by recent developments in approximate dynamic programming \cite{rusmevichientong2017dynamic,ma2018constant}, we consider the following modification to the residual function.

For each agent $t$, let $\ssx_t$ denote the probability that she is accepted by the prophet,
and let $\ssy_t$ denote her expected valuation conditional on being accepted.
We then define $\pi_t$ as follows, using backward induction over $t=T,T-1,\ldots,1$:
\begin{align} \label{eqn::defPi}
\pi_t=\sum_{t'>t:\{t,t'\}\in E}\ssx_{t'}\cdot\max\{\ssy_{t'}-\pi_{t'},0\}.
\end{align}
$\pi_t$ can be interpreted as the ``cost'' of accepting agent $t$ with respect to the graph-independence constraints.
Indeed, for all agents $t'>t$, the summand in (\ref{eqn::defPi}) is the already-computed ``surplus'' earned by the prophet on agent $t'$, and the sum is over all \textit{future} surpluses $t'$ which are blocked by accepting agent $t$.

Our modified residual function is based on a ``restricted'' prophet, who sees valuations $\hV_1,\ldots,\hV_T$ that have been reduced in two ways.
First, the restricted prophet only sees a non-zero valuation for an agent $t$ if the actual prophet would have accepted $t$ on that sample path (this is formalized in Section~\ref{sec::proofMR}); otherwise, the agent's valuation is zero.
Second, the valuation of every agent $t$ is further reduced by $\pi_t$, with $\pi_t$ as defined in (\ref{eqn::defPi}).
Our restricted residual function is then
\begin{align} \label{eqn::ourResidual}
\hR(Y)=\bE_{\hV_1,\ldots,\hV_T}\left[\max_{S:S\cup Y\in\cI}\sum_{t\in S}\hV_t\right],
\end{align}
and we define threshold $\tau(t|Y)=\frac{1}{d_1(M)+1}(\hR(Y)-\hR(Y\cup\{t\}))$.
Our algorithm accepts an agent $t$ if and only if she is both feasible and satisfies
\begin{align} \label{eqn::combinedThreshold}
V_t\ge\tau(t|Y)+\pi_t.
\end{align}

Returning to Example~\ref{ex::KWBad}, we would have $\pi_2=\ldots=\pi_T=0$ (because agents $t\ge2$ do not block any future agents), and $\pi_1\approx T-1$.
In this case, the matroid is free (because all the constraints are captured by the graph), so $\tau(t|Y)=0$ and our algorithm ends up accepting every agent $t$ using decision rule~(\ref{eqn::combinedThreshold}), which is the optimal control for Example~\ref{ex::KWBad}.

In general, $\tau(t|Y)$ represents our price for the matroid and $\pi_t$ represents our price for the graph.
$\pi_t$
\textit{discriminates} based on the agent $t$, looking at which agents $t'>t$ get blocked, but is \textit{static} in that it does not depend on the current state $Y$.
By contrast,
$\tau(t|Y)$ \textit{dynamically} considers the addition of element $t$ to the current $Y$,
but otherwise does not discriminate based on the agent $t$.
A further contribution of our work is that
we show how both $\tau(t|Y)$ and $\pi_t$ can be \textit{computed efficiently} when $G$ is induced by an intersection of interval-scheduling constraints of the form (\ref{eqn::intervalGraph}), by implementing an \textit{ex-ante relaxation} (Section~\ref{sec::computation}).

Finally, we describe our analysis, which consists of two steps.  First, we show that the algorithm earns at least $\frac{1}{d_1(M)+1}\hR(\emptyset)$, where $\hR(\emptyset)$ represents the welfare of the restricted prophet (Lemma~\ref{lem::firstStep}).
Lemma~\ref{lem::firstStep} differs from the original matroid prophet inequality in that the algorithm is further constrained by the graph, but gets to play against a prophet who sees valuations $\hV_t$ which have been reduced by $\pi_t$.
If the matroid is free and thus $d_1(M)=0$, then Lemma~\ref{lem::firstStep} is still non-trivial, as it says that the algorithm, constrained by the graph, can match the restricted prophet in welfare.
Our analysis concludes by showing that the restricted prophet earns at least $\frac{1}{d_2(G)+1}$ times the welfare of the actual prophet (Lemma~\ref{lem::secondStep}).

A general take-away from our paper is that the way in which constraints are \textit{modeled} can lead to different algorithms and prophet inequalities.
For example, simple constraints on item supplies can be modeled either with a partition matroid or by adding edges to our graph, which results in substantially different algorithms.
In general, is there a systematic way of dividing up constraints between feasibility structures to yield the best prophet inequality?
We leave this open as interesting future work.

\subsection{Other Related Work}

Prophet inequalities originated in \cite{krengel1977semiamarts,krengel1978semiamarts}, and the connection to posted-price mechanism design was discovered in \cite{chawla2010multi}.
There has since been a surge of literature on prophet inequalities, and we defer a complete literature review to the survey by Lucier \cite{lucier2017economic}.
Our work can be classified as having a fixed (adversarial) arrival order,
which can be contrasted with random-order prophet inequalities \cite{ehsani2018prophet};
general but structured feasible sets,
which can be contrasted with arbitrary feasible sets \cite{rubinstein2016beyond}
or refined results on rank-1 matroids \cite{correa2019recent};
and additive rewards, a special case of combinatorial rewards \cite{rubinstein2017combinatorial}.
Our paper is most related to the existing work involving matroids \cite{kleinberg2012matroid,dutting2015polymatroid,duetting2017prophet} and XOS valuation functions \cite{feldman2014combinatorial,duetting2017prophet}.
We should mention that interval-scheduling constraints have also been studied in \cite{im2011secretary,chawla2019pricing}, where it is shown that with no assumptions on the intervals, the guarantee relative to the prophet is at most $O(\log\log L/\log L)$, where $L$ is the length of (number of items in) the longest interval.  That is, with no assumptions on the intervals, a constant-factor is impossible.

Finally, we discuss two recent developments in approximate dynamic programming (ADP) from which we borrow techniques.
\cite{rusmevichientong2017dynamic} has developed an ADP-based algorithm which is within 1/2 of the optimal DP in an application with reusable resources.  This is the motivation behind our interval-scheduling constraints of the form (\ref{eqn::intervalGraph}).  \cite{ma2018constant} has established a guarantee of $1/(d+1)$ in a setting where each item uses up to $d$ resources.

Our work makes further contributions beyond these existing results in three ways.
First and most importantly, we show how to include ADP-based thresholds in the matroid residual function and analyze feasible sets defined by the intersection of a matroid and a graph.
Second, we unify the two existing ADP results by abstracting them using a graph, which leads to a more general result---we can allow for items to use \textit{multiple} (up to $d$) resources, each for a \textit{different} duration.
Finally, we extend their guarantees to be relative to the prophet (instead of the optimal DP), and also show how they can be applied on combinatorial auctions (instead of assortment optimization).

\section{Proof of Theorem~\ref{thm::mr}} \label{sec::proofMR}

We first summarize and formalize the notation and definitions from the Introduction, for the basic setting in Theorem~\ref{thm::mr}.

%\

\noindent\textbf{Online Selection Problem.} There is a ground set of agents $N=\{1,\ldots,T\}$ with valuations $V_1,\ldots,V_T$ drawn \textit{independently} from marginal distributions $F_1,\ldots,F_T$.
There is a matroid $M=(N,\cI)$ defined on the ground set, where $\cI$ is a collection of subsets of $N$ satisfying:
(i) $\emptyset\in\cI$;
(ii) if $S\in\cI$ and $S'\subseteq S$ then $S'\in\cI$; and
(iii) for $S,S'\in\cI$ with $|S|>|S'|$, there exists $t\in S\setminus S'$ such that $S'\cup\{t\}\in\cI$
(we refer to \cite{korte2006combinatorial} for more background on matroids and their use in optimization).
There is also a graph $G=(N,E)$ defined on $N$, where $E$ is a collection of size-2 subsets of $N$.
We let $\cF$ denote the collection of feasible sets, where a set of agents $S$ is feasible if it is both independent in the matroid (i.e.\ $S\in\cI$) and independent in the graph (i.e.\ $\{t,t'\}\notin E$ for all $t,t'\in S$).
The goal is to accept a max-value feasible set of agents as compared to a prophet.

%\

\noindent\textbf{Prophet.}
The prophet chooses a joint valuation distribution over $\bR^T$ with marginals $F_1,\ldots,F_T$.
On every realization, he sees the valuations and then selects a feasible set of agents.
Let $\ssx_t$ denote the probability that agent $t$ is selected, and let $\ssy_t$ denote her expected valuation conditional on being selected.  Let $\OPT$ denote the prophet's expected welfare, which equals $\sum_{t=1}^T\ssy_t\ssx_t$, by the linearity of expectation.
Furthermore, since on every realization, the prophet's selection must be independent in both the matroid and the graph, the vector $\ssx$ satisfies
\begin{align}
\sum_{t\in S}\ssx_t &\le\rank(S) &\forall\ S\subseteq N \label{constr::matroid} \\
\sum_{t\in S}\ssx_t &\le\alpha(G[S]) &\forall\ S=\{t'<t:\{t,t'\}\in E\},\text{ for some $t\in N$} \label{constr::graph}
\end{align}
where $\rank(S)=\max_{S'\subseteq S,S'\in\cI}|S'|$ denotes the max-cardinality matroid independent set contained within $S$, and $\alpha(G[S])$ denotes the max-cardinality graph independent set contained within $S$.
Note that we have relaxed constraints~(\ref{constr::graph}) to only a specific family of sets $S$, which are sufficient for Theorem~\ref{thm::mr} and computationally simpler.

%\

\noindent\textbf{Dynamic Programming Coefficients.}
Having defined $\ssx_t$ and $\ssy_t$, we compute $\pi_t=\sum_{t'>t:\{t,t'\}\in E}\ssx_{t'}[\ssy_{t'}-\pi_{t'}]^+$ by backward induction over $t=T,T-1,\ldots,1$, as in (\ref{eqn::defPi}).  We use $[\cdot]^+$ to denote the operator $\max\{\cdot,0\}$.

We note that such a backward-induction computation is only possible because we have assumed that the arrival order $1,\ldots,T$ is known in advance.

%\

\noindent\textbf{Restricted Prophet.}
The restricted prophet sees valuations $\vhV=(\hV_1,\ldots,\hV_T)$ drawn according to a joint distribution $\hD$ defined as follows.
First, an independent set $\hI$ in the matroid (which need not be independent in the graph) is randomly selected in a way such that $\Pr[t\in\hI]=\ssx_t$ for all $t\in N$ (this is possible because $\ssx$ lies in the matroid polytope defined by (\ref{constr::matroid})---we elaborate on this in Section~\ref{sec::computation}).
The restricted prophet then sees $\hV_t=\ssy_t-\pi_t$ if $t\in\hI$, and $\hV_t=-\pi_t$ otherwise.
The residual function~(\ref{eqn::ourResidual}) based on the restricted prophet is $\hR(Y)=\bE_{\vhV\sim\hD}[\max_{S:S\cup Y\in\cI}\sum_{t\in S}\hV_t]$.  We note that if $Y\notin\cI$, then $\hR(Y)=-\infty$.

We note that $\hR(\emptyset)=\sum_{t=1}^T\ssx_t[\ssy_t-\pi_t]^+$.
This is because on every realization of $\vhV$, the optimal $S$ to take is the set of agents $t$ with $\hV_t>0$, which is guaranteed to be independent in the matroid (since all such agents must have had $\hV_t=\ssy_t-\pi_t$).
A corollary
is that if the graph is empty and
$\pi_t=0$ for all $t$, then the restricted prophet earns $\sum_{t=1}^T\ssy_t\ssx_t$, matching the welfare of the actual prophet despite seeing ``binarized'' valuations $\hV_t$.
This reduction was introduced in \cite{lee2018optimal}.

%\

\noindent\textbf{Algorithm.}
The algorithm, having already accepted agents in $Y$, accepts an agent $t$ if and only if $Y\cup\{t\}$ is independent in the graph and $V_t\ge\tau(t|Y)+\pi_t$, as defined in (\ref{eqn::combinedThreshold}).
Note that $\tau(t|Y)=\frac{1}{d_1(M)+1}(\hR(Y)-\hR(Y\cup\{t\}))$, and we do not need to explicitly check that $Y\cup\{t\}$ is independent in the matroid, because if not, then $\tau(t|Y)=\infty$.
Let $\ALG$ denote the expected welfare of this algorithm.

%\

We note that this algorithm requires computing the prophet's values of $\ssx_t,\ssy_t$, and evaluating expectations over the restricted prophet's correlated distribution $\hD$.
We ignore computational issues in this section and discuss how the algorithm can be implemented via an \textit{ex-ante relaxation} in Section~\ref{sec::computation}.

We now establish Theorem~\ref{thm::mr}, the conceptual result that the gap between an online algorithm and any prophet is at most $(d_1(M)+1)(d_2(G)+1)$ (with $d_1(M),d_2(G)$ as defined in Definition~\ref{def::d1}--\ref{def::d2}), via
a sequence of two lemmas.

\begin{lemma}  \label{lem::firstStep}
The algorithm earns at least $\frac{1}{d_1(M)+1}$ times the welfare of the restricted prophet.
That is, $\ALG\ge\frac{1}{d_1(M)+1}\hR(\emptyset)$.
\end{lemma}

\begin{proof}
Let $Y$ denote the random set of agents accepted at the end of the algorithm, and for all $t=1,\ldots,T$, let $Y_t$ denote $Y\cap\{1,\ldots,t\}$, the set of agents accepted up to and including agent $t$.
The algorithm's expected welfare equals
\begin{align}
\ALG  &= \bE\left[\sum_{t \in Y} (V_t - \tau(t|Y_{t-1})) + \sum_{t \in Y}\tau(t|Y_{t-1})\right] \nonumber \\
&= \bE\left[\sum_{t \in Y} (V_t - \tau(t|Y_{t-1}))+ \frac{1}{d_1(M)+1} \sum_{t\in Y}(\hR(Y_{t-1})-\hR(Y_{t-1}\cup\{t\}))\right] \nonumber \\
&=  \frac{1}{d_1(M)+1} \hR(\emptyset)  + \bE\left[\sum_{t \in Y} (V_t - \tau(t|Y_{t-1}))\right] - \frac{1}{d_1(M)+1} \bE[\hR(Y)] \label{eqn::algDecomp}
\end{align}
where the second equality follows from the definition of $\tau$, and the third equality follows from the fact that $Y_{t-1}\cup\{t\}=Y_t$ for all $t\in Y$, causing the latter sum to telescope.

We now upper-bound the negative term from (\ref{eqn::algDecomp}) in the following proposition.
It mostly follows from existing results \cite{kleinberg2012matroid,lee2018optimal}, so its proof is deferred to the appendix.
It relies on the submodularity of the matroid residual function.
Since our restricted residual function considers a prophet who is only constrained by the matroid, our function is also submodular.
\begin{proposition} \label{prop::existing}
%\begin{align*}
$\frac{1}{d_1(M)+1}\hR(Y)\le\sum_{t=1}^T\ssx_t[\ssy_t-\tau(t|Y_{t-1})-\pi_t]^+$.
%\end{align*}
\end{proposition}

In the meantime, we would like to lower-bound the second term from (\ref{eqn::algDecomp}), which represents the ``surplus'' earned by the algorithm beyond the thresholds $\tau(t|Y_{t-1})$, while being constrained by both matroid and graph independence.
The following proposition is novel and crucial to our analysis.

\begin{proposition} \label{prop::key}
%\begin{align*}
$\bE\left[\sum_{t \in Y} (V_t - \tau(t|Y_{t-1}))\right]\ge\bE\left[\sum_{t=1}^T\ssx_t[\ssy_t-\tau(t|Y_{t-1})-\pi_t]^+\right]$.
%\end{align*}
\end{proposition}

\begin{proof}[of Proposition~\ref{prop::key}]
We decompose the LHS as $\bE\left[\sum_{t\in Y}(V_t-\tau(t|Y_{t-1})-\pi_t)\right]+\bE\left[\sum_{t\in Y}\pi_t\right]$ and analyze the two expectations separately.

The first expectation can be re-written as
\begin{align} \label{eqn::2378}
\bE\left[\sum_{t\in Y}(V_t-\tau(t|Y_{t-1})-\pi_t)\right]
&=\sum_{t=1}^T\bE_{Y_{t-1}}\Big[\bE_{V_t}[\bI(t\in Y)\cdot(V_t-\tau(t|Y_{t-1})-\pi_t)|Y_{t-1}]\Big]
\end{align}
after using both the linearity of expectation and the tower property of conditional expectation.
Now, recall that as agent $t$ arrives, she is accepted if and only if she is feasible (in both the matroid and graph), and $V_t-\tau(t|Y_{t-1})-\pi_t\ge0$.
If $Y_{t-1}\cup\{t\}$ does not form an independent set in the matroid, then $\tau(t|Y_{t-1})=\infty$, since $\hR(Y_{t-1}\cup\{t\})$ is understood to equal $-\infty$ when the maximization problem in the residual is infeasible.
Therefore, we can write
\begin{align*}
\bI(t\in Y)\cdot(V_t-\tau(t|Y_{t-1})-\pi_t)=\Feas^G(Y_{t-1}\cup\{t\})\cdot[V_t-\tau(t|Y_{t-1})-\pi_t]^+,
\end{align*}
where $\Feas^G(Y_{t-1}\cup\{t\})$ is the indicator random variable for $Y_{t-1}\cup\{t\}$ forming an independent set in the graph.
Making this substitution for every agent $t$ on the RHS of (\ref{eqn::2378}), we get that it equals
\begin{align} \label{eqn::firstTerm}
\sum_{t=1}^T\bE_{Y_{t-1}}\Big[\Feas^G(Y_{t-1}\cup\{t\})\cdot\bE_{V_t}\big[[V_t-\tau(t|Y_{t-1})-\pi_t]^+\big]\Big],
\end{align}
where we have used the fact that $V_t$ is independent from $Y_{t-1}$.

Meanwhile, the second expectation can be re-written as
\begin{align}
\bE\left[\sum_{t\in Y}\pi_t\right]
&=\bE\left[\sum_{t\in Y}\sum_{t'>t:\{t,t'\}\in E}\ssx_{t'}[\ssy_{t'}-\pi_{t'}]^+\right] \nonumber \\
&=\bE\left[\sum_{t'=1}^T\ssx_{t'}[\ssy_{t'}-\pi_{t'}]^+\Big(\sum_{t<t':\{t,t'\}\in E}\bI(t\in Y)\Big)\right] \nonumber \\
&=\bE\left[\sum_{t'=1}^T\ssx_{t'}[\ssy_{t'}-\pi_{t'}]^+(1-\Feas^G(Y_{t'-1}\cup\{t'\}))\right]
\label{eqn::secondTerm}
\end{align}
where the first equality applies the definition of $\pi_t$ from (\ref{eqn::defPi}),
the second equality switches sums,
and the third equality holds because agent $t'$ forms an independent set with $Y_{t'-1}$ in the graph if and only if none of its neighbors $t<t'$ have been accepted into $Y$.

Now, note that the sum of expressions \eqref{eqn::firstTerm} and \eqref{eqn::secondTerm} can be rewritten as
\begin{align*}
&\sum_{t=1}^T\bE_{Y_{t-1}}\Big[\Feas^G(Y_{t-1}\cup\{t\})\cdot\bE_{V_t}\big[[V_t-\tau(t|Y_{t-1})-\pi_t]^+\big]\Big] \\
&+\sum_{t=1}^T\bE_{Y_{t-1}}\Big[(1-\Feas^G(Y_{t-1}\cup\{t\}))\cdot\ssx_{t}[\ssy_{t}-\pi_{t}]^+\Big],
\end{align*}
which is lower-bounded by
\begin{align*}
\sum_{t=1}^T\bE_{Y_{t-1}}\left[\min\Big\{\bE_{V_t}\big[[V_t-\tau(t|Y_{t-1})-\pi_t]^+\big],\ssx_t\cdot[\ssy_t-\pi_t]^+\Big\}\right].
\end{align*}
We would like to further lower-bound both of the terms inside the $\min\{\cdot\}$ operator, by $\ssx_t\cdot[\ssy_t-\tau(t|Y_{t-1})-\pi_t]^+$.
For the second term, this is obvious, since the thresholds $\tau(t|Y_{t-1})$ are non-negative.
For the first term, note that $V_t$ takes an average value of $\ssy_t$ on an $\ssx_t$-fraction of sample paths.
Hence by Jensen's inequality, the expectation over $V_t$ is at least $\ssx_t[\ssy_t-\tau(t|Y_{t-1})-\pi_t]^+$ (the $[\cdot]^+$ operator is convex).

Returning to the fact that the sum of expressions \eqref{eqn::firstTerm} and \eqref{eqn::secondTerm} was the LHS of the original inequality, this completes the proof of Proposition~\ref{prop::key}. $\blacksquare$
\end{proof}

Equipped with Propositions~\ref{prop::existing}--\ref{prop::key}, the proof of Lemma~\ref{lem::firstStep} now follows from (\ref{eqn::algDecomp}).  Indeed, taking an expectation over $Y$ on both sides in the result of Proposition~\ref{prop::existing}, (\ref{eqn::algDecomp}) implies that $\ALG\ge\frac{1}{d_1(M)+1}\hR(\emptyset)$, which is the desired result. $\blacksquare$
\end{proof}

\begin{lemma} \label{lem::secondStep}
The restricted prophet earns at least $\frac{1}{d_2(G)+1}$ times the welfare of the actual prophet.
That is, $\hR(\emptyset)\ge\frac{1}{d_2(G)+1}\cdot\OPT$.
\end{lemma}

\begin{proof}
Recall that $\hR(\emptyset)=\sum_{t=1}^T\ssx_t[\ssy_t-\pi_t]^+$.
%it suffices to show that $\sum_{t=1}^T\ssx_t[\ssy_t-\pi_t]^+\ge\frac{1}{d_2(G)+1}\cdot\OPT$.
We apply the definition of $\pi_t$ from (\ref{eqn::defPi}) and switch sums to derive that
\begin{align*}
\sum_{t=1}^T\ssx_t[\ssy_t-\pi_t]^+
&\ge\sum_{t=1}^T\ssx_t\ssy_t-\sum_{t=1}^T\ssx_t\pi_t \\
&=\sum_{t=1}^T\ssx_t\ssy_t-\sum_{t=1}^T\ssx_t\sum_{t'>t:\{t,t'\}\in E}\ssx_{t'}[\ssy_{t'}-\pi_{t'}]^+ \\
&=\sum_{t=1}^T\ssx_t\ssy_t-\sum_{t'=1}^T\ssx_{t'}[\ssy_{t'}-\pi_{t'}]^+\Big(\sum_{t<t':\{t,t'\}\in E}\ssx_t\Big).
\end{align*}
Now, using the fact that the prophet's values of $\ssx_t$ satisfy (\ref{constr::graph}),
%and taking $S=\{t<t':\{t,t'\}\in E\}$,
the sum in parentheses is upper-bounded by $\alpha(G[\{t<t':\{t,t'\}\in E\}])$ for all $t'=1,\ldots,T$.
By Definition~\ref{def::d2}, all of these values are upper-bounded by $d_2(G)$.  Therefore, since $\ssx_{t'}[\ssy_{t'}-\pi_{t'}]^+\ge0$, we have that
\begin{align*}
\sum_{t=1}^T\ssx_t[\ssy_t-\pi_t]^+\ge\sum_{t=1}^T\ssx_t\ssy_t-d_2(G)\sum_{t'=1}^T\ssx_{t'}[\ssy_{t'}-\pi_{t'}]^+,
\end{align*}
and rearranging yields $\sum_{t=1}^T\ssx_t[\ssy_t-\pi_t]^+\ge\frac{1}{d_2(G)+1}\sum_{t=1}^T\ssx_t\ssy_t=\frac{\OPT}{d_2(G)+1}$. $\blacksquare$
\end{proof}

\subsection{Computing and Constructing the Prophet's Distribution via an Ex-ante Relaxation} \label{sec::computation}

In this section we establish computational efficiency
assuming that the graph $G$ is induced by $d$-dimensional interval-scheduling constraints.
We use an \textit{ex-ante} relaxation defined by an LP, and establish three facts:
\begin{enumerate}
\item Theorem~\ref{thm::mr} still holds if we replace the prophet with this ex-ante relaxation, resulting in a guarantee of $\frac{1}{(d_1(M)+1)(d+1)}$ (because $d_2(G)\le d$); \label{fact::1}
\item Our algorithm based on this ex-ante relaxation is computationally efficient; \label{fact::2}
\item The ex-ante relaxation upper-bounds the welfare of any prophet. \label{fact::3}
\end{enumerate}
These facts together show that a computationally-efficient algorithm can earn at least $\frac{1}{(d_1(M)+1)(d+1)}$ times the welfare of any prophet.

%First we define $d$-dimensional interval-scheduling constraints.

\begin{definition}[$d$-dimensional Interval-scheduling Constraints] \label{def::dDimIntervalConstraints}
The agents arrive at times $1,\ldots,T$ to be served by $J$ different resources.
Each agent $t$ requests the attention of up to $d$ resources, for different durations of time starting from $t$.
Formally, associated with agent $t$ are intervals $\{\bbI^j_t=[t,u^j_t]:j\in U_t\}$, where $U_t$ is a set of at most $d$ resources, with $u^j_t\ge t$ for all $j\in U_t$.

An agent $t$ can be served only if all of the resources in $U_t$ are available.
Thus, a set of agents is feasible only if their requested intervals are disjoint for \textbf{every} resource.
Agents $t,t'$ are adjacent in the graph if $\bbI^j_t\cap\bbI^j_{t'}\neq\emptyset$ for \textbf{any} $j$.
\end{definition}

\begin{definition}[Discrete Valuations]
We assume that the marginal valuations are input as discrete distributions.
That is, they are supported over a finite set of $K$ values $v^1,\ldots,v^K\in\bR$, and for each agent $t$, we let $p^k_t\ge0$ denote the probability that $V_t=v^k$ for every $k=1,\ldots,K$, with $\sum_kp^k_t=1$.
\end{definition}

\begin{definition}[Ex-ante Relaxation]
The ex-ante relaxation is defined by the following LP.
\begin{align}
\max\sum_{t=1}^T\sum_{k=1}^Kv^kx_{tk} \nonumber \\
%\sum_{k=1}^Kx_{tk} &=x_t &\forall\ t=1,\ldots,T \\
\sum_{t\in S}\Big(\sum_{k=1}^Kx_{tk}\Big) &\le\rank(S) &\forall\ S\subseteq N \label{lp::matroid} \\
\sum_{t'\le t:j\in U_{t'},\bbI^j_{t'}\ni t}\Big(\sum_{k=1}^Kx_{t'k}\Big) &\le1 &\forall\ t\in N;j\in U_t \label{lp::interval} \\
0\le x_{tk} &\le p^k_t &\forall\ t\in N;k=1,\ldots,K \nonumber
\end{align}
We then consider the values of $x_{tk}$ from an optimal LP solution and define
\begin{align*}
\ssx_t=\sum_{k=1}^Kx_{tk},\ \ssy_t=\frac{1}{\ssx_t}\sum_{k=1}^Kv^kx_{tk}
\end{align*}
for all agents $t$ (where $\ssy_t=0$ if $\ssx_t=0$).
\end{definition}
In the LP, variable $x_{tk}$ can be interpreted as the probability that agent $t$ has valuation $v^k$ and is accepted into the feasible set.  Note that in an optimal solution, $\ssy_t$ will equal the average value of $V_t$ on its top $\ssx_t$ quantile.

We now formalize the three facts stated above, with
the proofs deferred to the appendix.
The second fact references classical results in combinatorial optimization about separation \cite{grotschel1981ellipsoid} and rounding \cite{cunningham1984testing} for the matroid polytope.

\begin{proposition}[fact~\ref{fact::1} from above] \label{prop::factOne}
Our algorithm, when defined based on the values of $\ssx_t,\ssy_t$ from the ex-ante relaxation, has welfare at least $\frac{1}{(d_1(M)+1)(d+1)}\sum_{t=1}^T\ssy_t\ssx_t$.
\end{proposition}

\begin{proposition}[fact~\ref{fact::2} from above] \label{prop::factTwo}
Assuming oracle access to the matroid rank function, the values of $\ssx_t,\ssy_t$ from the ex-ante relaxation can be efficiently computed.  Furthermore, the restricted prophet's correlated distribution $\hD$ has a compact representation which can be efficiently computed.
\end{proposition}

\begin{proposition}[fact~\ref{fact::3} from above] \label{prop::factThree}
The expected welfare of any prophet, who can choose the correlation between $V_1,\ldots,V_T$ and select a feasible $S\in\cF$ maximizing $\sum_{t\in S}V_t$ on every realization, is upper-bounded by the optimal LP value of $\sum_{t=1}^T\ssy_t\ssx_t$.
\end{proposition}

%We now present the main result of Section~\ref{sec::computation}.  Having established computational efficiency,
%\begin{theorem}
%Suppose that the valuation distributions are discrete and that the graph $G$ is defined by the intersection of interval constraints.
%Suppose that every agent requires at most $d$ resources, i.e.\ $|\{j:I^j_t\neq\emptyset\}|\le d$ for all $t\in N$.
%Finally, suppose that we are given oracle access to the rank function of the matroid $M$.
%Then a polynomial-time algorithm can earn expected welfare at least $\frac{1}{2(d+1)}$ times that of any prophet.
%\end{theorem}

\section{Generalization to XOS Combinatorial Auctions} \label{sec::generalizations}

We generalize our result to a setting where each agent has a random valuation function over a set of items, and the graph and matroid constraints are defined on the items. An agent's valuation function is realized upon arrival, and the set of items allocated to the agent must then be decided.
Specifically, there are $T$ agents and a set of items $N$.
There is a matroid $(N, \cI)$ and a graph $(N, E)$ defined over these items, and the total set of items allocated must be both independent in the matroid and the graph.

We assume that $N$ is partitioned into sets $N_1, \dots, N_T$ such that agent $t$ can only be given items in $N_t$.
This does not lose generality, in that we can always replace an item in $N$ with a clique in the graph over $N_1\cup\ldots\cup N_T$, whose vertices are all adjacent to the same neighbors.  The copies of the item forming a clique ensures that at most one of them can be taken, while not increasing the constant $d_2(G)$ for the graph (because they all block each other).

% for $t\neq t'$, two items $i \in N_{t}$ and $i' \in N_{t'}$ which were supposed to correspond to the same item could be made adjacent in the graph $G$, so that only one of $i$ and $i'$ can be taken.  We have simply replaced each vertex in the original graph with a clique, all of which is connected to the same neighbors.
%However, this could change the value of the constant $d_2(G)$, which affects our welfare guarantee.

% We make the assumption that all edges in the graph are \textit{across} time. That is, if $\{i, i'\} \in E$ for $i \in N_{t}$ and $i' \in N_{t'}$, it must be that $t \neq t'$. This assumption implies that the graph constraints do not restrict agent $t$ from being allocated any subset of $N_t$. If we wanted to impose such a constraint, we can modify the valuation function for agent $t$ accordingly. For example, if we wanted to disallow agent $t$ from taking both of $i, j \in N_t$, then for any $S \subseteq N_t \setminus \{i, j\}$ we can modify the valuation function to be $v_t(S \cup \{i\} \cup \{j\}) = \max\{v_t(S \cup \{i\} ), v_t(S \cup \{i\})\}$.

An allocation is denoted by $(Y_1, \dots, Y_T) \subseteq \prod_{t=1}^T 2^{N_t}$, and let $\cY_t = (Y_1, \dots, Y_t)$. Sometimes we abuse notation to assume $\cY_t = Y_1 \cup \dots \cup Y_t$. The set of feasible allocations is
\[
\cF = \{(Y_1, \dots, Y_T) \subseteq \prod_t 2^{N_t}: Y_1 \cup \dots \cup Y_T \text{ independent in $M$ and $G$}\}.
\]

We have similar definitions of $d_1(M)$ and $d_2(M)$ as before. $d_1(M)=0$ if $\cI=2^N$, and $d_1(M)=1$ otherwise.
% \begin{align} \label{eqn::dGraphDef::XOS}
$d_2(G)=\max_{i\in N_t: t \in [T]}\alpha\big(G\big[\{i' \in N_{t'}:\{i,i'\}\in E, t' < t\}\big]\big)$,
% \end{align}
where $G[\cdot]$ denotes the subgraph of $G$ induced by a set of vertices, and $\alpha(\cdot)$ denotes the maximum size of an independent set in a graph.

Each agent $t$ has a random valuation function $v_t: 2^{N_t} \rightarrow \bR_{\geq0}$ drawn from a known distribution. At time $t$, agent $t$'s valuation function realizes independently to a valuation function $v^k_t$ with probability $p_t^k$ for $k=1,\dots,K$. Then, the set of items $Y_t \subseteq N_t$ allocated to the agent is decided.
We require the all valuation functions to be fractionally-subadditive, i.e.\ XOS.
\begin{definition}
  For a set of items $N_t$, a valuation function $v_t:2^{N_t} \rightarrow \bR$ is \textit{XOS}, or fractionally-subadditive, if it can be written as a maximum of a collection of additive functions: $$v_t(S)=\max_{\ell\in[L]}v_{\ell}(S) \quad\;\forall\ S\subseteq N_t,$$
  % \begin{align*}
  % % v^k(S)=\max_{j\in\{1,\ldots,m\}}\sum_{i\in S}v_{i,j}^k &&\forall\ S\subseteq N
  % v(S)=\max_{\ell\in\{1,\ldots,L\}}v_{\ell}(S) &&\forall\ S\subseteq N
  % \end{align*}
  where each $v_{\ell}$ is an additive\footnote{A valuation function $v$ is \textit{additive} if $v(S)+v(S') = v(S\cup S')$ for all disjoint subsets of items $S$ and $S'$.} valuation function.
\end{definition}

We now formalize the prophet and algorithm.

\noindent\textbf{Prophet.}
As before, we compare our online algorithm to a prophet who is able to choose an arbitrary correlated distribution $\ssD$ over $v_1,\ldots,v_T$.
Furthermore, for every realization of the valuation functions $\cv = (v_1, \dots, v_T)$, we define the prophet's allocation to be $\proph(\cv)=(\proph(\cv)_1,\ldots,\proph(\cv)_T)$, which always satisfies
\[
\proph(\cv) \in \argmax_{(Y_1, \dots, Y_T) \in \cF}\sumt v_t(Y_t).
\]
Let $x_{t}^{k*}(S) = \Pr_{\cv \sim \ssD}(\proph(\cv)_t = S | v_t = v_t^k)$ be the conditional probability that the prophet chooses to allocate $S \subseteq N_t$ to agent $t$ given that $v_t$ realizes to $v_{t}^k$. Then, by linearity of expectation, the prophet's value is
\begin{align}
  \OPT =
%\bE_{\cv \sim \ssD}[\proph(\cv)] = 
\sumt \sumk p_t^k \sumS x_t^{k*}(S) v_t^k(S). \label{eq::OPT::XOS}
\end{align}

Since the prophet's allocation must be independent in the graph on every realization, for every $S = S_1\cup  \dots \cup S_T \subseteq N$, the following must be satisfied:
\begin{align*}
\sumt \sum_{J_t \subseteq N_t} \bI[\text{$J_t$ allocated}] \sum_{i \in J_t} \bI[i \in S_t] \leq \alpha(G[S]).
\end{align*}
That is, the number of items in $S$ that is taken must not be greater than the size of the biggest independent set in the graph induced by $S$.
By taking expectations, we get:
\begin{align}
% \sumt \sumk p_t^k \sum_{J_t \subseteq N_t} x_t^{k*}(J_t) \sum_{i \in J_t} \bI[i \in S_t] \leq \rank(S)  \label{XOS::matroidrank}\\
\sumt \sumk p_t^k \sum_{J_t \subseteq N_t} x_t^{k*}(J_t) \sum_{i \in J_t} \bI[i \in S_t] \leq \alpha(G[S]). \label{XOS::graphrank}
% \sumt  \sum_{i \in S_t} \sumk p_t^k \sum_{J_t \subseteq N_t: i \in J_t} x_t^{k*}(J_t) \leq \alpha(G[S]).
\end{align}

%\

\noindent\textbf{XOS Decomposition.}
%As before, we simulate a restricted prophet to compute $\tau$. 
%For $S\subseteq N$ and a realization $\cv$, the restricted prophet takes a subset of items of $\proph(\cv)$ that is independent in the matroid with $S$ and maximizes welfare. 
%As before, we compute costs $\pi_t(i)$ as the cost of allocating item $i$ to agent $t$.
%We use a decomposition of an XOS valuation by item.
For $S \subseteq N_t$, consider the valuation $v_t^k(S)$. Since $v_t^k$ is XOS, it can be written as $v_t^k(S) = \max_{\ell \in [L]} w_\ell(S)$, where each $w_\ell: 2^{N_t} \rightarrow \bR$ is an additive function.
% it is a maximum of additive functions, say $w_1, \dots, w_L: 2^{N_t} \rightarrow \bR$.
Let $u_t^k(i, S)$ be the value of item $i$ for the additive function that \textit{supports} $v_t^k(S)$. That is, if $v_t^k(S) = w_\ell(S)$, then let $u_t^k(i, S) = w_\ell(\{i\})$ for all $i \in S$.
 % i.e. $v_t^k(S) = \sum_{i \in S} u_t^k(i, S)$.
%$v_t^k(S)=\sum_{i \in S}u_t^k(i, S)$, and hence the values $u_t^k(i, S)$ is a decomposition of $v_t^k(S)$.
% $u_t^k(i, S)$ can be thought of as the value added by item $i$ from the valuation $v_t^k(S)$.
% Because the valuation function $v_t^k$ is XOS, the following property holds:

The following property holds for any XOS function $v_t^k$:
\begin{align} \label{eq::XOSproperty}
v_t^k(J) \geq \sum_{i \in J} u_t^k(i, S) \quad \forall J \subseteq S \subseteq N_t,
\end{align}
since the additive function that supports $S$ has exactly the value $\sum_{i \in J} u_t^k(i, S)$, so $v_t^k(J)$ can only be higher.  This property is also used in \cite{feldman2014combinatorial}.
%This property is used in a crucial step in the proof and it does not hold for general subadditive valuations (see .

%\

\noindent\textbf{Dynamic Programming Coefficients.}
Define $\pi_t(i)$ and $\hu_t^k(i,S)$ recursively using backwards induction over all $t=T, T-1, \dots, 1$, $S\subseteq N_t$, and $i \in N_t$:
%as follows:
\begin{align}
  \pi_t(i) &= \sum_{t' > t} \sum_{k=1}^K p_{t'}^k \sum_{S'\subseteq N_{t'}} x_{t'}^{k*}(S') \sum_{i' \in S'} \hu_{t'}^k(i', S') \bI(\{i,i'\} \in E). \label{eqn::defPi::XOS} \\
  \hu_t^k(i, S) &= [u_t^k(i, S) - \pi_t(i)]^+. \nonumber
\end{align}
% Define $\pi_t(i)$ recursively using backwards induction over $t=T, T-1, \dots, 1$ and $i \in N$ as follows:
% \begin{align*}
%   \pi_t(i) &= \sum_{t' > t} \sum_{k=1}^K p_{t'}^k \sum_{S'\subseteq N_{t'}} x_{t'}^{k*}(S') \sum_{i' \in S'} \bI(\{i,i'\} \in E)[u_{t'}^k(i', S') - \pi_{t'}(i)]^+ )
% \end{align*}
Analogous to \eqref{eqn::defPi}, $\pi_t(i)$ is the ``cost" of allocating item $i$ at time $t$. We sum over all future surpluses $\hu^k_{t'}(i',S')$ that are ``blocked" by $G$ if item $i$ is taken, where the future surpluses are pre-computed and depend not only on $i'$ but \textit{also} on $S'$.
This is because we need to separately account for the sets $S'$ which could be allocated to agent $t'$, for each possible realization of valuation function $v_{t'}$.
%when her valuation function realizes to $v^k_{t'}$.
%, we compare whether $u^k_{t'}(i',S')$ is greater than $\pi_{t'}(i')$.
%$\hu_t^k(i,S)$ is the ``restricted" reward that is used in the definition of the restricted prophet.

%\

\noindent\textbf{Restricted Prophet.}
For $S \subseteq N$, we define the \textit{restricted residual} of $S$ as:
\begin{align*}
\hR(S) =  \bE_{\cv\sim\ssD}[\max_{\substack{\cJ \subseteq \proph(\cv) \\ S \cup \cJ \in \cI }} \sumt \sum_{i \in J_t} \hu_t(i, \proph(\cv)_t)].
\end{align*}
where $\cJ = (J_1, \dots, J_T)$, and $\cJ \subseteq \proph(\cv)$ means $J_t \subseteq \proph(\cv)_t \; \forall t$.
That is, from the prophet solution $\proph(\cv)$, we take the subset of it which is feasible with $S$ and maximizes the restricted rewards.
If $S \notin \cI$, $\hR(S)$ is defined as $-\infty$.
If $S = \emptyset$, then setting $\cJ=\proph(\cv)$ is a maximizer for every $\cv$ (because $\proph(\cv)\in\cI$), and hence
%$\cJ = \proph(\cv)$ and
$\hR(\emptyset)$ equals $\sumt \sumk p_t^k \sumS x_t^{k*}(S)\sum_{i \in S} \hu_t^k(i, S)$, which we refer to as the
value of the \textit{restricted prophet}.

%\

\noindent\textbf{Algorithm.}
Let $\tau(S | \cY_{t-1}) = \frac{1}{d_1(M)+1}(R(\cY_{t-1}) - R(\cY_{t-1} \cup S))$ denote the threshold for the matroid at time $t$, when the set of items taken so far is $\cY_{t-1}$. Note that $\tau(S | \cY_{t-1}) = \infty$ if $\cY_{t-1} \cup S \notin \cI$.
At time $t$, if the agent's valuation realizes to $v_t^k$, the algorithm allocates subset
\begin{align}
Y_t \in \argmax_{\substack{S \subseteq N_t: \\ \text{$S$ feasible} }} \big(v_t^k(S) - \sum_{i \in S} \pi_t(i) - \tau(S | \cY_{t-1})\big). \label{alg::XOS}
\end{align}
That is, the algorithm allocates the feasible subset which has the largest ``surplus" over the sum of the two thresholds.

We are now ready to state our generalization of Theorem~\ref{thm::mr}.

\begin{theorem} \label{thm::XOS}
% Let $N = N_1 \cup \dots \cup N_T$ be a set of items, with $N_t \cap N_{t'} = \emptyset$ for $t \neq t'$.
Let $N_1, \dots, N_T$ be a collection of disjoint items, and let $N = N_1 \cup \dots \cup N_T$.
Let $M = (N, \cI)$ be a matroid and $G = (N, E)$ be a graph.
If every realization of the random valuation function $v_t:2^{N_t} \rightarrow \bR$ is XOS for all $t \in [T]$, then the expected welfare of an online algorithm is at least
\begin{align*}
\frac{1}{(d_1(M)+1)(d_2(G)+1)}\cdot\OPT,
\end{align*}
where $\OPT$ is the expected welfare of a prophet who can choose the correlation between $v_1,\ldots,v_T$ and see their realizations beforehand.
\end{theorem}

The proof of Theorem~\ref{thm::XOS} is deferred to Appendix~\ref{sec::thm2pf}.
Below, we also generalize Definition~\ref{def::dDimIntervalConstraints}, and bound $d_2(G)$ in the setting where every item
$i \in N$ requests the attention of up to $d$ resources for a duration of time.
When combined, these results imply a $\frac{1}{(d_1(M)+1)(d+1)}$-guarantee.
\begin{proposition} \label{prop::d2<=d}
  Let $J$ be a set of resources, and let $U_i \subseteq J$ for every $i \in N$ with $|U_i| \leq d$. Associated with every item $i \in N_t$ are intervals $\{\mathbb{I}_{i}^j = [t, u_{i}^j]: j \in U_i\}$ with $u_i^j \geq t$ for all $j \in U_i$. An item $i \in N_t$ can be allocated at time $t$ only if all resources in $U_i$ are available. Then, this constraint can be modelled by a graph $G$ that satisfies $d_2(G) \leq d$.
\end{proposition}

%We can apply this theorem to several settings.
%\begin{example}
%Let $M$ be a set of resources, and let $A_i \subseteq M$ be the set of resources that item $i \in N$ consumes, where $|A_i| \leq d$. We cannot allocate two items that use the same resource. Let $G$ be the graph where there is an edge between $i$ and $j$ whenever $A_i \cap A_j \neq \emptyset$. For every item $i \in N$, the neighbours of $i$ in the graph can be partitioned into $|A_i|$ groups, each partition representing the resource that is shared. At most one item from each of these partitions can be in an independent set; therefore $d_2(G) \leq d$, and the online algorithm gurantees $\frac{1}{d+1} \OPT$.
%\end{example}
%
%\begin{example}
%Extending the previous example, suppose each agent only required items for a finite usage time. Then, starting with the graph $G$ from the previous example, we can remove edges between items across agents that will not interfere with each other. Since we are only removing edges, we still have $d_2(G) \leq d$.
%\end{example}

%\

\bibliographystyle{splncs04}
\bibliography{bibliography}

\begin{thebibliography}{10}
\providecommand{\url}[1]{\texttt{#1}}
\providecommand{\urlprefix}{URL }
\providecommand{\doi}[1]{https://doi.org/#1}

\bibitem{chawla2010multi}
Chawla, S., Hartline, J.D., Malec, D.L., Sivan, B.: Multi-parameter mechanism
  design and sequential posted pricing. In: Proceedings of the forty-second ACM
  symposium on Theory of computing. pp. 311--320. ACM (2010)

\bibitem{chawla2019pricing}
Chawla, S., Miller, J.B., Teng, Y.: Pricing for online resource allocation:
  intervals and paths. In: Proceedings of the Thirtieth Annual ACM-SIAM
  Symposium on Discrete Algorithms. pp. 1962--1981. Society for Industrial and
  Applied Mathematics (2019)

\bibitem{correa2019recent}
Correa, J., Foncea, P., Hoeksma, R., Oosterwijk, T., Vredeveld, T.: Recent
  developments in prophet inequalities. ACM SIGecom Exchanges  \textbf{17}(1),
  61--70 (2019)

\bibitem{cunningham1984testing}
Cunningham, W.H.: Testing membership in matroid polyhedra. Journal of
  Combinatorial Theory, Series B  \textbf{36}(2),  161--188 (1984)

\bibitem{duetting2017prophet}
D{\"u}etting, P., Feldman, M., Kesselheim, T., Lucier, B.: Prophet inequalities
  made easy: Stochastic optimization by pricing non-stochastic inputs. In: 2017
  IEEE 58th Annual Symposium on Foundations of Computer Science (FOCS). pp.
  540--551. IEEE (2017)

\bibitem{dutting2015polymatroid}
D{\"u}tting, P., Kleinberg, R.: Polymatroid prophet inequalities. In:
  Algorithms-ESA 2015, pp. 437--449. Springer (2015)

\bibitem{ehsani2018prophet}
Ehsani, S., Hajiaghayi, M., Kesselheim, T., Singla, S.: Prophet secretary for
  combinatorial auctions and matroids. In: Proceedings of the Twenty-Ninth
  Annual ACM-SIAM Symposium on Discrete Algorithms. pp. 700--714. SIAM (2018)

\bibitem{feldman2014combinatorial}
Feldman, M., Gravin, N., Lucier, B.: Combinatorial auctions via posted prices.
  In: Proceedings of the twenty-sixth annual ACM-SIAM symposium on Discrete
  algorithms. pp. 123--135. SIAM (2014)

\bibitem{grotschel1981ellipsoid}
Gr{\"o}tschel, M., Lov{\'a}sz, L., Schrijver, A.: The ellipsoid method and its
  consequences in combinatorial optimization. Combinatorica  \textbf{1}(2),
  169--197 (1981)

\bibitem{im2011secretary}
Im, S., Wang, Y.: Secretary problems: Laminar matroid and interval scheduling.
  In: Proceedings of the twenty-second annual ACM-SIAM symposium on Discrete
  Algorithms. pp. 1265--1274. Society for Industrial and Applied Mathematics
  (2011)

\bibitem{kleinberg2012matroid}
Kleinberg, R., Weinberg, S.M.: Matroid prophet inequalities. In: Proceedings of
  the forty-fourth annual ACM symposium on Theory of computing. pp. 123--136.
  ACM (2012)

\bibitem{korte2006combinatorial}
Korte, B., Vygen, J.: Combinatorial optimization, Third edition. Springer
  (2006)

\bibitem{krengel1977semiamarts}
Krengel, U., Sucheston, L.: Semiamarts and finite values. Bulletin of the
  American Mathematical Society  \textbf{83}(4),  745--747 (1977)

\bibitem{krengel1978semiamarts}
Krengel, U., Sucheston, L.: On semiamarts, amarts, and processes with finite
  value. Probability on Banach spaces  \textbf{4},  197--266 (1978)

\bibitem{lee2018optimal}
Lee, E., Singla, S.: Optimal online contention resolution schemes via ex-ante
  prophet inequalities. arXiv preprint arXiv:1806.09251  (2018)

\bibitem{lucier2017economic}
Lucier, B.: An economic view of prophet inequalities. ACM SIGecom Exchanges
  \textbf{16}(1),  24--47 (2017)

\bibitem{ma2018constant}
Ma, Y., Rusmevichientong, P., Sumida, M., Topaloglu, H.: A constant-factor
  approximation algorithm for network revenue management  (2018)

\bibitem{rubinstein2016beyond}
Rubinstein, A.: Beyond matroids: Secretary problem and prophet inequality with
  general constraints. In: Proceedings of the forty-eighth annual ACM symposium
  on Theory of Computing. pp. 324--332. ACM (2016)

\bibitem{rubinstein2017combinatorial}
Rubinstein, A., Singla, S.: Combinatorial prophet inequalities. In: Proceedings
  of the Twenty-Eighth Annual ACM-SIAM Symposium on Discrete Algorithms. pp.
  1671--1687. SIAM (2017)

\bibitem{rusmevichientong2017dynamic}
Rusmevichientong, P., Sumida, M., Topaloglu, H.: Dynamic assortment
  optimization for reusable products with random usage durations. Draft
  available at https://people. orie. cornell.
  edu/huseyin/publications/reusable\_rm. pdf  (2017)

\end{thebibliography}

\appendix

\section{Why weakly balanced prices do not exist for Example~\ref{ex::KWBad}} \label{appx::balancedPricesDNE}

We explain why even the most general framework from \cite{duetting2017prophet}, which seeks \textit{weakly balanced prices} in a deterministic setting, does not appear to yield a constant-factor guarantee for Example~\ref{ex::KWBad}.
The reasoning is similar to that described in Section~\ref{sec::techniques}, with the issue caused by the optimum ``starting over''.

Consider Theorem~3.2 from \cite{duetting2017prophet} and consider any constants $\alpha$, $\beta_1$, and $\beta_2$.
We construct 
values of $T$, $C$, and $\eps$ for which their algorithm based on balanced prices extracts an arbitrary-small fraction of welfare.
We will follow the notation from Section~3 of \cite{duetting2017prophet}.

For every agent $t$, the corresponding outcome space is $\{\emptyset,\acc\}$, where $\acc$ refers to agent $t$ being accepted while $\emptyset$ refers to agent $t$ being rejected.
Let $\bx$ be the allocation in which only the second agent is accepted. Suppose $\bv$ is the valuation profile where the first agent has valuation $\frac{C + T\eps}{\eps}$, which occurs with probability $\eps$. In this case, $\bv(\ALG(\bv)) = \frac{C + T\eps}{\eps}$, since the prophet only allocates the item to the first agent. Then, $\mathcal{F}_\bx$ (the exchange-compatible set) cannot contain any allocation $y$ that accepts the first agent, because if it did, then it would not satisfy $(y_1, \bx_{-1}) \in \mathcal{F}$.
Therefore, it must be that $\bv(\OPT(\bv, \mathcal{F}_\bx)) \leq T-1$. Then, for $\bx$ and $\bv$ to satisfy the first constraint in weakly balanced prices, it must be that
\begin{align*}
\sum_{i \in T} p_i(x_i | \bx_{[i-1]}) = p_2(\acc | \emptyset ) \geq \frac{1}{\alpha} (\bv(\ALG(\bv))-\bv(\OPT(\bv, \mathcal{F}_\bx))) \geq \frac{1}{\alpha}(\frac{C}{\eps} + 1).
\end{align*}
Their posted price mechanism uses prices $\delta\cdot\bE_{\bv}[p_i^{\bv}(x_i | y)]$ for $\delta=\frac{1}{\beta_1 + \max\{2\beta_2, 1/\alpha\}}$, and the aforementioned valuation function realizes with probability $\eps$. Therefore, the price for agent 2 in the posted price mechanism is greater than
\[
\frac{\delta C}{\alpha} + \frac{\eps \delta}{\alpha}  > 1
\]
if $C \geq \frac{\alpha}{\delta}$. Similarly, the price for agents $i=3, \dots, T$ will also be greater than 1, so those agents will never be accepted. Then, the posted-price mechanism will achieve welfare $C + T\eps$, whereas the prophet will achieve $C+T-1+\eps$. As $T \rightarrow \infty$ and $\eps = o(1/T)$, the fraction of welfare achieved by the posted-price mechanism goes to 0.

\section{Deferred Proofs}

\begin{proof}[of Proposition~\ref{prop::existing}]
The residual function $\hR$, which involves the prophet selecting a max-value independent set in a matroid, is submodular by \cite{kleinberg2012matroid}.  That is, for all subsets $S$ and $S'$, $\hR(S\cup S')+\hR(S\cap S')\le\hR(S)+\hR(S')$.  Applying this inequality repeatedly yields the following:
\begin{align}
\hR(Y)
&=\bE_{\vhV}\left[\max_{S:S\cup Y\in\cI}\sum_{t\in S}(\hV_t-\tau(t|Y_{t-1})+\tau(t|Y_{t-1}))\right] \nonumber \\
&=\bE_{\vhV}\left[\max_{S:S\cup Y\in\cI}\left(\sum_{t\in S}(\hV_t-\tau(t|Y_{t-1}))+\frac{1}{d_1(M)+1}\sum_{t\in S}(\hR(Y_{t-1})-\hR(Y_{t-1}\cup\{t\}))\right)\right] \nonumber \\
&\le\bE_{\vhV}\left[\max_{S:S\cup Y\in\cI}\left(\sum_{t\in S}(\hV_t-\tau(t|Y_{t-1}))+\frac{1}{d_1(M)+1}\sum_{t\in S}(\hR(Y)-\hR(Y\cup\{t\}))\right)\right] \nonumber \\
&\le\bE_{\vhV}\left[\max_{S:S\cup Y\in\cI}\left(\sum_{t\in S}(\hV_t-\tau(t|Y_{t-1}))+\frac{1}{d_1(M)+1}(\hR(Y)-\hR(S\cup Y))\right)\right]. \label{eqn::2170}
\end{align}
Now, $\hR(Y)-\hR(S\cup Y)\le d_1(M)\hR(Y)$, because if $M$ is the free matroid then $\hR(Y)=\hR(S\cup Y)$ and $d_1(M)=0$; otherwise, $-\hR(S\cup Y)\le0$ and $d_1(M)=1$.
Therefore, we can rearrange (\ref{eqn::2170}) to get
\begin{align} \label{eqn::2789}
\hR(Y)-\frac{d_1(M)}{d_1(M)+1}\hR(Y)\le\bE_{\vhV}\left[\max_{S:S\cup Y\in\cI}\sum_{t\in S}(\hV_t-\tau(t|Y_{t-1}))\right].
\end{align}
Now, recall that each $\hV_t$ takes value $\ssy_t-\pi_t$ with probability $\ssx_t$, and value $-\pi_t$ otherwise.  Without the constraint that $S\cup Y\in\cI$, the RHS of (\ref{eqn::2789}) would equal $\sum_{t=1}^T\ssx_t[\ssy_t-\pi_t-\tau(t|Y_{t-1})]^+$.
Therefore, the RHS of (\ref{eqn::2789}) can be at most this value, completing the proof. $\blacksquare$
\end{proof}

\begin{proof}[of Proposition~\ref{prop::factOne}]
Our algorithm and analysis from Theorem~\ref{thm::mr} only required that the values of $\ssx_t$ satisfy (\ref{constr::matroid})--(\ref{constr::graph}), and that $\ssy_t$ corresponds to the average value of $V_t$ on an $\ssx_t$-fraction of sample paths.  We verify that all of this is still true for the ex-ante's values of $\ssx_t,\ssy_t$.

(\ref{constr::matroid}) follows immediately from LP constraint~(\ref{lp::matroid}) and the definition of $\ssx_t$.

To check (\ref{constr::graph}), consider any $t=1,\ldots,T$ and $S=\{t'<t:\{t,t'\}\in E\}$.
An agent $t'<t$ can only be adjacent to $t$ if there exists $j\in U_t\cap U_{t'}$ such that interval $\bbI^j_{t'}\ni t$, i.e.\ agent $t'$ lapses time $t$ on resource $j$.
Therefore, $S$ can be partitioned into $|U_t|$ sets where the agents $t'$ in each set all lapse time $t$ on some resource $j$.
By LP constraint~(\ref{lp::interval}), the sum of values of $\ssx_{t'}$ for the agents $t'$ in the same set is at most 1.
Therefore, $\sum_{t'\in S}\ssx_{t'}\le|U_t|$.
Since $|U_t|\le d$ for all agents $t$, constraints~(\ref{constr::graph}) hold with the RHS replaced by $d$.
The proof of Theorem~\ref{thm::mr} (specifically, that of Lemma~\ref{lem::secondStep}) holds with $\max_t\alpha(G[\{t'<t:\{t,t'\}\in E\}])=d_2(G)$ replaced by $d$.

Finally, $\ssy_t$ equals the average value of $V_t$ on its top $\ssx_t$ quantile, which can be directly checked from its definition and the fact that $x_{tk}$ was an optimal LP solution.  Therefore, we can invoke Theorem~\ref{thm::mr} to complete the proof. $\blacksquare$
\end{proof}

\begin{proof}[of Proposition~\ref{prop::factTwo}]
The LP is polynomially-sized except for the exponential family of constraints~(\ref{lp::matroid}).
This family of constraints define the matroid polytope and can be efficiently separated over, assuming oracle access to the matroid rank function (note that this is a submodular function minimization problem).
Furthermore, using the GLS ellipsoid method, separation implies that the LP can be solved to optimality and hence the vectors $\ssx,\ssy$ can be computed.
For further background on these results, we refer to \cite[Sec.~14.3]{korte2006combinatorial} and \cite{grotschel1981ellipsoid}.

Since $\ssx$ lies in the matroid polytope defined by (\ref{constr::matroid}), which is integral \cite[Sec.~13.4]{korte2006combinatorial}, $\ssx$ can indeed be represented by a distribution over independent sets.
Furthermore, there are explicit rounding procedures for doing so, which can compute a small convex combination of independent sets equaling $\ssx$ assuming an oracle to the matroid \cite{cunningham1984testing} (note that a convex representation with at most $T+1$ sets exists, by Caratheodory's theorem \cite[Sec.~3]{korte2006combinatorial}).
Therefore, the restricted prophet's distribution $\hD$ has small support and can be explicitly constructed in polynomial time, which allows us to efficiently evaluate the expectation over $\vhV$ in the definition of the residual $\hR$, and hence efficiently run our algorithm. $\blacksquare$
\end{proof}

\begin{proof}[of Proposition~\ref{prop::factThree}]
Consider any correlated distribution for $V_1,\ldots,V_T$ and prophet selection rule.  Set $x_{tk}$ to be the probability that the prophet accepts agent $t$ when her valuation realizes to $v^k$, for all $t$ and $k$.  We argue that this forms a feasible solution to the LP.  The constraints (\ref{lp::matroid})--(\ref{lp::interval}) follow from the linearity of expectation, since the prophet must select a set that is independent in both the matroid and the graph on every realization (note that the edges in the graph are defined so that that the interval constraints~(\ref{lp::interval}) are indeed satisfied by the $\{0,1\}$-incidence vector of any independent set in the graph).
Meanwhile, $x_{tk}\le p^k_t$ because the marginal probability that the valuation of agent $t$ is $v^k$ is at most $p^k_t$.
Finally, the objective value of the LP equals the expected welfare of the prophet.
Since the prophet corresponds to a feasible LP solution, the optimal LP solution can be no less, completing the proof.
 $\blacksquare$
\end{proof}

\begin{proof}[of Proposition~\ref{prop::d2<=d}]
  Let $G =(N, E)$ be the graph where $\{i, i'\} \in E$ if $\mathbb{I}_{i}^j \cap \mathbb{I}_{i'}^j \neq \emptyset$ for any resource $j$.
  % If there exists an item $i' \in N_{t'}$ with $\mathbb{I}_{i}^j \cap \mathbb{I}_{i'}^j \neq \emptyset$ and $t' < t$ that was allocated, then $i$ cannot be allocated.
  Then, an item $i \in N_t$ can be allocated if and only if none of its neighbors are allocated.
  Therefore, a set of items is independent in the graph if and only if it is a feasible allocation.

  We show $d_2(G) \leq d$. Consider any $i \in N_t$, and let $S=\{i' \in N_{t'}: \{i, i'\} \in E, t' < t\}$ be the neighbors of $i$ that come before time $t$ that ``blocks`` $i$. We must show that the largest independent set in the graph $G[S]$ contains at most $d$ nodes. Since every neighbor of $i$ uses at least one resource in common with $i$, we can partition $S$ into $|U_i|$ sets based on which resource they have in common. (If a neighbor has more than one resource in common with $i$, then choose one of the resources at random.)
  Each set in the partition form a clique in the graph, since they all share a resource in common and their intervals overlap at time $t$.
  Therefore, an independent set in the graph $G[S]$ can contain at most one item from each of the sets in the partition. Thus, $\alpha(G[S]) \leq |U_i| \leq d$. $\blacksquare$
\end{proof}

\section{Proof of Theorem~\ref{thm::XOS}} \label{sec::thm2pf}

The proof structure is similar to that of Theorem~\ref{thm::mr}. Lemmas~\ref{lem::firstStep::XOS} and \ref{lem::secondStep::XOS} are analogous to Lemmas~\ref{lem::firstStep} and \ref{lem::secondStep}, respectively.  Propositions~\ref{prop::existing::XOS} and \ref{prop::key::XOS} also correspond to Propositions~\ref{prop::existing} and \ref{prop::key}.

\begin{lemma} \label{lem::firstStep::XOS}
  The algorithm earns at least $\frac{1}{d_1(M)+1}$ times the welfare of the restricted prophet. That is, $\ALG \geq \frac{1}{d_1(M)+1} \hR(\emptyset)$.
\end{lemma}

\begin{proof}
Let $Y_t \subseteq N_t$ be the random variable corresponding to the set of items allocated to agent $t$ by the algorithm, and let $\cY_t = Y_1 \cup \dots \cup Y_t$. The algorithm's expected welfare equals

\begin{align}
\ALG  &= \bE\left[\sumt (v_t(Y_t) - \tau(Y_t|\cY_{t-1})) + \sumt\tau(Y_t|\cY_{t-1})\right] \nonumber \\
&= \bE\left[\sumt (v_t(Y_t) - \tau(Y_t|\cY_{t-1}))+ \frac{1}{d_1(M)+1} \sumt(\hR(\cY_{t-1})-\hR(\cY_{t-1}\cup Y_t))\right] \nonumber \\
&=  \frac{1}{d_1(M)+1} \hR(\emptyset)  + \bE\left[\sumt (v_t(Y_t) - \tau(Y_t|\cY_{t-1}))\right] - \frac{1}{d_1(M)+1} \bE[\hR(\cY_T)], \label{eqn::algDecomp::XOS}
\end{align}
where the second equality follows from the definition of $\tau$, and the third equality follows from the fact that $\cY_{t-1}\cup Y_t=\cY_t$ for all $t$, causing the latter sum to telescope.

\begin{proposition} \label{prop::existing::XOS}
%\begin{align*}
% $\frac{1}{d_1(M)+1}\hR(\cY_T)\le\sum_{t=1}^T\ssx_t[\ssy_t-\tau(t|Y_{t-1})-\pi_t]^+$.
%\end{align*}=
\[
\frac{1}{d_1(M)+1} \hR(\cY_T) \leq \sumt \sumk p_t^k \sumS x_t^{k*}(S) \max_{J \subseteq S} (\sum_{i \in J} \hu_t^k(i, S) - \tau(J| \cY_{t-1})).
\]
\end{proposition}
% Lemma~\ref{lem::existing::XOS} places an upper bound on the negative term from (\ref{eqn::algDecomp::XOS}).
% It mostly follows from existing results \cite{kleinberg2012matroid,lee2018optimal}, so its proof is deferred to the appendix.
% It relies on the submodularity of the matroid residual function.

\begin{proof}[of Proposition~\ref{prop::existing::XOS}]

% Fix $\cY_T$.

\begin{align}
\hR(\cY_T)  &= \bE_{\cv\sim\ssD}[\max_{\substack{\cJ \subseteq \proph(\cv) \\ \cY_T \cup \cJ \in \cI }} \sumt \sum_{i \in J_t} \hu_t^k(i, \proph(\cv)_t)] \nonumber \\
&= \bE_{\cv\sim\ssD}[\max_{\substack{\cJ \subseteq \proph(\cv) \\ \cY_T \cup \cJ \in \cI }} \Big(\sumt (\sum_{i \in J_t} \hu_t^k(i, \proph(\cv)_t) - \tau(J_t| \cY_{t-1})) + \sumt \tau(J_t|\cY_{t-1})\Big)]  \label{eq::Y_T2::XOS}
\end{align}

Fix any $\cJ = (J_1 , \dots ,J_T)$ such that $\cJ \cup \cY_T \in \cI$. The residual function $\hR$ is submodular by \cite{kleinberg2012matroid}.
That is, for all subsets $S$ and $S'$, $\hR(S\cup S')+\hR(S\cap S')\le\hR(S)+\hR(S')$.  Applying this inequality repeatedly yields the following:

\begin{align}
\sumt \tau(J_t | \cY_{t-1}) &= \frac{1}{d_1(M)+1}  \sumt (\hR(\cY_{t-1}) - \hR(Y_{t-1} \cup J_t)) \nonumber \\
&\leq \frac{1}{d_1(M)+1} \sumt (\hR(\cY_T  \cup J_{t-1}) - \hR(\cY_T \cup J_t))  \nonumber \\
&\leq \frac{1}{d_1(M)+1} \sumt (\hR(\cY_T \cup J_1 \cup \dots \cup J_{t-1}) - \hR(\cY_T \cup J_1 \cup \dots \cup J_t))  \nonumber \\
&\leq \frac{1}{d_1(M)+1} (\hR(\cY_T) -\hR(\cY_T \cup \cJ)). \label{eq::sumthresholds}
\end{align}
Now, $\hR(\cY_T) -\hR(\cY_T \cup \cJ)\le d_1(M)\hR(\cY_T)$, because if $M$ is the free matroid then $\hR(\cY_T)=\hR(\cY_T \cup \cJ)$ and $d_1(M)=0$; otherwise, $-\hR(\cY_T \cup \cJ)\le0$ and $d_1(M)=1$.
Therefore, we substitute \eqref{eq::sumthresholds} into \eqref{eq::Y_T2::XOS} and rearrange to get
\begin{align*}
\hR(Y)-\frac{d_1(M)}{d_1(M)+1}\hR(Y)&\le\bE_{\cv\sim\ssD}[\max_{\substack{\cJ \subseteq \proph(\cv) \\ \cY_T \cup \cJ \in \cI }} \sumt (\sum_{i \in J_t} \hu_t^k(i, \proph(\cv)_t) - \tau(J_t| \cY_{t-1})) ] \\
\frac{1}{d_1(M)+1}\hR(Y) &\leq \bE_{\cv\sim\ssD}[ \sumt \max_{J_t \subseteq \proph(\cv)_t}(\sum_{i \in J_t} \hu_t^k(i, \proph(\cv)_t) - \tau(J_t| \cY_{t-1})) ] \\
&\leq \sumt \sumk p_t^k \sumS x_t^{k*}(S) \max_{J \subseteq S} (\sum_{i \in J} \hu_t^k(i, S) - \tau(J| \cY_{t-1})).
\end{align*}
The second step follows from loosening the constraint on $\cJ$ under the $\max$ operator, and the final step opens up the expectation using $p_t^k = \Pr(v_t = v_t^k)$ and $x_t^{k*}(S) = \Pr_{\cv\sim\ssD}(\OPT(v)_t = S | v_t = v_t^k)$. $\blacksquare$

\end{proof}

\begin{proposition} \label{prop::key::XOS}
\[
\bE[\sumt (v_t(Y_t) - \tau(Y_t | \cY_{t-1}))] \geq
\bE[\sumt \sumk p_t^k \sumS x_t^{k*}(S) \max_{J \subseteq S} (\sum_{i \in J} \hu_t^k(i, S) - \tau(J| \cY_{t-1})) ].
\]
%\end{align*}
\end{proposition}
% Lemma~\ref{lem::key::XOS} lower-bounds the ``surplus'' earned by the algorithm beyond the thresholds $\tau(Y_t|\cY_{t-1})$, and needs to consider that the algorithm is constrained by both matroid and graph independence.
% It is novel and crucial to our analysis.

\begin{proof}[of Proposition~\ref{prop::key::XOS}]
It is useful to define a mapping $\phi_{t, \cY_{t-1}}^k: 2^{N_t} \rightarrow 2^{N_t}$ as:
% For a subset of items $S \subseteq N_t$ that is independent in both the matroid and the graph, define the mapping
\[\phi_{t, \cY_{t-1}}^k(S) =
\begin{cases}
      \argmax_{J\subseteq S} (\sum_{i \in J} \hu_t^k(i, S) - \tau(J| \cY_{t-1})) & S\in \cF \\
      \emptyset &  S \notin \cF
   \end{cases}
   % \argmax_{J\subseteq S} (\sum_{i \in J} \hu_t^k(i, S) - \tau(J| \cY_{t-1})).
\]
$\phi_{t, \cY_{t-1}}^k(S)$ represents the subset of $S$ which maximizes the ``surplus" at time $t$ and valuation $k$, when $\cY_{t-1}$ are the items that have been allocated. Note that $S$ is not necessarily feasible when $\cY_{t-1}$ has already been allocated, but $\phi_{t, \cY_{t-1}}^k(S) \cup \cY_{t-1}$ is independent in $M$ at time $t$, since $\tau(J|\cY_{t-1})=-\infty$ if $J \cup \cY_{t-1}$ is not independent in $M$.
Define
\[
\bar{\phi}_{t, \cY_{t-1}}^k(S) = \{i \in \phi_{t, \cY_{t-1}}^k(S): i \cup \cY_{t-1} \text{ independent in $G$}, u_t^k(i, S) \geq \pi_t(i)\}
\]
to be the subset of $\phi_{t, \cY_{t-1}}^k(S)$ that is independent in $G$ and whose value is above the threshold $\pi_t(i)$.

We decompose the LHS as $\bE[\sumt (v_t(Y_t) - \sum_{i \in Y_t} \pi_t(i) - \tau(Y_t | \cY_{t-1})) ] + \bE[\sumt \sum_{i \in Y_t} \pi_t(i)]$
and analyze the two expectations separately.

We use the definition of the algorithm \eqref{alg::XOS} to re-write the first expectation:

\begin{align}
  &\bE[\sumt (v_t(Y_t) - \sum_{i \in Y_t} \pi_t(i) - \tau(Y_t | \cY_{t-1})) ] \nonumber \\
  = & \bE[\sumt \sumk p_t^k \max_{S \subseteq N_t: S \text{ feasible}} (v_t^k(S) - \sum_{i \in S} \pi_t(i) - \tau(S|\cY_{t-1}))]. \label{eq::algexpect}
\end{align}

For any $S$, since $\bar{\phi}_{t, \cY_{t-1}}^k(S)$ is feasible, the algorithm has a higher surplus than had $\bar{\phi}_{t, \cY_{t-1}}^k(S)$ been taken:

\begin{align} \label{eq::lemalg::XOS}
& \max_{S \subseteq N_t: S \text{ feasible}} (v_t^k(S) - \sum_{i \in S} \pi_t(i) - \tau(S|\cY_{t-1} ))  \nonumber \\
&\geq v_{t}^{k}(\bar{\phi}_{t, \cY_{t-1}}^k(S)) - \sum_{i \in \bar{\phi}_{t, \cY_{t-1}}^k(S)} \pi_t(i) - \tau(\bar{\phi}_{t, \cY_{t-1}}^k(S) | \cY_{t-1}).
\end{align}

The property \eqref{eq::XOSproperty} of XOS valuations gives us $v_t^k(\bar{\phi}_{t, \cY_{t-1}}^k(S)) \geq \sum_{i \in \bar{\phi}_{t, \cY_{t-1}}^k(S)} u_t^k(i, S)$, since $\bar{\phi}_{t, \cY_{t-1}}^k(S) \subseteq S$.
We use this property and take a convex combination of \eqref{eq::lemalg::XOS} over subsets $S \subseteq N_t$.
% ($\phi_{t, \cY_{t-1}}^k(S)$ is undefined for any $S \notin \cF$, but we also have $x_{t}^{k*}(S) = 0$ for such $S$, so it is not included in the convex combination.)
\eqref{eq::algexpect} becomes:

% By using this property and taking a convex combination of \eqref{eq::lemalg::XOS}, \eqref{eq::algexpect} becomes:

\begin{align}
  &\bE[\sumt (v_t(Y_t) - \sum_{i \in Y_t} \pi_t(i) - \tau(Y_t | \cY_{t-1})) ] \nonumber \\
  \geq& \bE[\sum_{t=1}^T \sum_{k=1}^K p_{t}^{k} \sum_{S\subseteq N_{t}} x_{t}^{k*}(S) \big(\sum_{i \in \bar{\phi}_{t, \cY_{t-1}}^k(S)} (u_{t}^k(i, S) - \pi_t(i) )- \tau(\bar{\phi}_{t, \cY_{t-1}}^k(S)| \cY_{t-1}) \big)]  \nonumber  \\
  \geq& \bE[\sum_{t=1}^T \sum_{k=1}^K p_{t}^{k} \sum_{S\subseteq N_{t}} x_t^{k*}(S) ( \sum_{i \in \phi_{t, \cY_{t-1}}^k(S)} \Feas^G(\cY_{t-1} \cup \{i\})\hu_{t}^{k}(i, S)  - \tau(\phi_{t, \cY_{t-1}}^k(S) | \cY_{t-1}) )]. \label{eq::feasible::XOS}
\end{align}

The last step uses the definition of $\bar{\phi}_{t, \cY_{t-1}}^k(S)$ and the fact that $\tau(\bar{\phi}_{t, \cY_{t-1}}^k(S) | \cY_{t-1}) \leq  \tau(\phi_{t, \cY_{t-1}}^k(S) | \cY_{t-1})$.
The second expectation can be rewritten as
\begin{align}
 \bE[\sumt \sum_{i \in Y_t} \pi_t(i) ]
&=\bE[ \sumt \sum_{i \in Y_t} \sum_{t' > t} \sum_{k'=1}^K p_{t'}^{k'} \sum_{S'\subseteq N_{t'}} x_{t'}^{k'*}(S') \sum_{i' \in S'} \hu_{t'}^{k'}(i', S') \bI(\{i,i'\} \in E) ]  \nonumber \\
&= \bE[ \sum_{t'=1}^T \sum_{k'=1}^K p_{t'}^{k'} \sum_{S'\subseteq N_{t'}} x_{t'}^{k'*}(S') \sum_{i' \in S'} \hu_{t'}^{k'}(i', S') [\sum_{t<t'}  \sum_{i \in Y_t}\bI(\{i,i'\} \in E)]],   \label{eq::infeasible::XOS}
% \geq& \bE[ \sum_{t'=1}^T \sum_{k'=1}^K p_{t'}^{k'} \sum_{S'\subseteq N_{t'}} x_{t'}^{k'*}(S') \sum_{i' \in S'} \hu_{t'}^{k'}(i', S') \bI(\text{$i'$ infeasible at time $t'$ in $G$})] \\
% \geq& \bE[ \sum_{t=1}^T \sum_{k=1}^K p_{t}^{k} \sum_{S\subseteq N_{t}} x_{t}^{k*}(S) \sum_{i \in \phi_{t, \cY_{t-1}}^k(S)} \hu_{t}^{k}(i, S) \bI(\text{$i$ infeasible in $G$}) ] \\
\end{align}
after applying the definition of $\pi_t(i)$ from \eqref{eqn::defPi::XOS} and switching sums. At time $t'$, $i'$ is dependent in $G$ if and only if there was an item $i$ with $\{i, i'\} \in E$ that was allocated at an earlier time step. Therefore, the term in the square brackets in \eqref{eq::infeasible::XOS} is greater than $1- \Feas^G(\cY_{t'-1} \cup i')$.

Recombining \eqref{eq::feasible::XOS} and \eqref{eq::infeasible::XOS}:
\begin{align*}
  & \bE[\sumt (v_t(Y_t) - \tau(Y_t | \cY_{t-1}))]    \\
\geq& \bE[\sumt \sumk p_{t}^{k} \sumS x_{t}^{k*}(S)( \sum_{i \in \phi_{t, \cY_{t-1}}^k(S)} \hu_{t}^{k}(i, S) -  \tau(\phi_{t, \cY_{t-1}}^k(S) | \cY_{t-1}) ) ]\\
% \geq& \bE[\sumt \sumk p_{t}^{k} \sumS x_{t}^{k*}(S)( \sum_{i \in \phi_{t, \cY_{t-1}}^k(S)} \hu_{t}^{k}(i, S) -  \tau(\phi_{t, \cY_{t-1}}^k(S) | \cY_{t-1}) )] \\
\geq& \bE[\sumt \sumk p_{t}^{k} \sumS x_{t}^{k*}(S) \max_{J\subset S} (\sum_{i \in  J} \hu_{t}^{k}(i, S) -  \tau(J | \cY_{t-1}) )].
\end{align*}

The last inequality uses the definition of $\phi_{t, \cY_{t-1}}^k(S)$ and the fact that $x_{t}^{k*}(S) = 0$ for all $S \notin \cF$, since the prophet cannot take such an $S$. $\blacksquare$

\end{proof}

Equipped with Propositions~\ref{prop::existing::XOS}--\ref{prop::key::XOS}, the proof of Lemma~\ref{lem::firstStep::XOS} now follows from (\ref{eqn::algDecomp::XOS}).  $\blacksquare$

\end{proof}

\begin{lemma} \label{lem::secondStep::XOS}
  The restricted prophet earns at least $\frac{1}{d_2(G)+1}$ times the welfare of the actual prophet.
That is, $\hR(\emptyset)\ge\frac{1}{d_2(G)+1}\cdot\OPT$.
\end{lemma}

\begin{proof}

By the definition of $\hR(\emptyset)$, we have
\small
\begin{align*}
  \hR(\emptyset)
  =& \sumt \sumk p_t^k \sumS x_t^{k*}(S) \sum_{i \in S} (u_t^k(i, S) - \pi_t(i))^+ \\
  \geq& \sumt \sumk p_t^k \sumS x_t^{k*}(S) \sum_{i \in S} u_t^k(i, S) - \sumt \sumk p_t^k \sumS x_t^{k*}(S) \sum_{i \in S}\pi_t(i) \\
  =& \OPT - \sumt \sumk p_t^k \sumS x_t^{k*}(S) \sum_{i \in S}\Big[\sum_{t' > t} \sum_{k'=1}^K p_{t'}^{k'} \sum_{S'\subseteq N_{t'}} x_{t'}^{{k'}*}(S') \sum_{i' \in S'} \hu_{t'}^{k'}(i', S') \bI(\{i,i'\} \in E))\Big], \\
\end{align*}
\normalsize
where the last step uses the definition of $\OPT$ and $\pi_t(i)$. Next, we move all the summations within the square brackets to the front:

\small
\begin{align*}
  \hR(\emptyset)
  \geq \OPT - \sum_{t'=1}^T \sum_{k'=1}^K p_{t'}^{k'} \sum_{S'\subseteq N_{t'}} x_{t'}^{{k'}*}(S') \sum_{i' \in S'} \hu_{t'}^k(i', S')[\sum_{t < t'} \sumk p_t^k \sumS x_t^{k*}(S) \sum_{i \in S} \bI(\{i,i'\} \in E))].
  % \geq& \bE[\OPT] - \sum_{t'=1}^T \sum_{k=1}^K p_{t'}^k \sum_{S'\subseteq N_{t'}} x_{t'}^{k*}(S') \sum_{i' \in S'} \hu_{t'}^k(i', S')[d_2(M)] \\
  % =& \frac{1}{d_2(M) + 1} \bE[\OPT]
\end{align*}
\normalsize

For a fixed $t'$ and $i'$, the prophet solution satisfies \eqref{XOS::graphrank} with $S = \{i \in N_{t}: \{i, i'\} \in E, t < t'\}$. Therefore, the term in the square brackets is less then $\alpha(G[S]) \leq d_2(G)$ by definition of $d_2(G)$. Then,

\begin{align*}
  \hR(\emptyset)
  &\geq \OPT - \sum_{t'=1}^T \sum_{k'=1}^K p_{t'}^{k'} \sum_{S'\subseteq N_{t'}} x_{t'}^{{k'}*}(S') \sum_{i' \in S'} \hu_{t'}^{k'}(i', S')[d_2(G)] \\
  &\geq \OPT - d_2(G) \hR(\emptyset).
\end{align*}
Rearranging yields $\hR(\emptyset) \geq \frac{1}{d_2(G)+1} \OPT$. $\blacksquare$

\end{proof}

\end{document}